%% file: main.tex
\documentclass[runningheads]{llncs}
\pdfoutput=1

\newif\iffinal
\finaltrue

\input{packages}

\input{commands}
\input{macros}

\title{Automatic and Incremental Repair\\ for Speculative Information Leaks}
\titlerunning{Automatic and Incremental Repair for Speculative Information Leaks}
\author{Joachim Bard\inst{1} \and Swen Jacobs\inst{1} \and Yakir Vizel\inst{2}}

\institute{
CISPA Helmholtz Center for Information Security, Germany\\
	\and
Technion, Israel\\
}
\authorrunning{}%

\begin{document}
\lstset{%
  language=C,
  basicstyle=\scriptsize\tt,
  columns=flexible,
  keepspaces=true,
  escapechar=@,
  commentstyle=\color{gray},
  numbers=left,
  numbersep=5pt,
  numberstyle=\tiny,
  keywordstyle=\bfseries\color{blue},
  morekeywords={bool}
}

\maketitle

\begin{abstract}
We present \FenceSynth, the first model-checking based framework 
for \emph{automatic repair} of programs with respect to information leaks 
in the presence of side-channels and speculative execution. 
\FenceSynth is based on formal models of attacker capabilities, including observable side channels, 
inspired by the \spectre attacks.
For a given attacker model, \FenceSynth is able to either
prove that the program is secure, or \emph{detect} potential side-channel vulnerabilities 
and \emph{automatically insert mitigations} such that the resulting code is provably secure.
Moreover, \FenceSynth can provide a \emph{certificate} for the security of the program that can be independently checked.
We have implemented \FenceSynth in the SeaHorn framework and show that it can effectively repair security-critical code, for example the AES encryption from the OpenSSL library.
\end{abstract}

\input{intro}

\input{prelims}

\input{speculative-semantics}

\input{repair}

\input{evaluation}
\input{related-work}

\bibliographystyle{splncs04}%
\bibliography{bibliography}

\newpage
\appendix
\input{appendix-proofs}

\input{appendix-evaluation}

\end{document}

%% file: packages.tex
\usepackage[utf8]{inputenc}
\usepackage{amsmath}
\usepackage{amssymb}
\usepackage{thmtools}
\usepackage{thm-restate}
\usepackage{mathtools}
\usepackage[english]{babel}
\usepackage{booktabs}
\usepackage{semantic}
\usepackage{caption}
\usepackage{subcaption}

\usepackage{sansmath}
\usepackage{algorithm2e}
\usepackage{comment}

\usepackage{cite}

\usepackage[dvipsnames]{xcolor}

\usepackage{paralist}

\usepackage{pgfplots}

\usepackage{tikz}
\usetikzlibrary{automata}
\usetikzlibrary{shapes,positioning}

\usepackage{todonotes}
\usepackage{listings}[mathescape=true]
\usepackage{xspace}

\PassOptionsToPackage{hyphens}{url}
\usepackage{hyperref}
\usepackage[capitalise]{cleveref}

%% file: commands.tex
\newcommand*{\jb}[1]{\todo[inline,backgroundcolor=green]{JB: #1}}
\newcommand*{\jbi}[1]{\textcolor{Green}{[JB: #1]}}
\newcommand*{\sj}[1]{\todo[inline,backgroundcolor=cyan]{SJ: #1}}
\newcommand*{\sji}[1]{\textcolor{blue}{[SJ: #1]}}
\newcommand*{\yv}[1]{\textcolor{red}{[YV: #1]}}

\iffinal
\renewcommand{\jbi}[1]{}
\renewcommand{\sji}[1]{}
\renewcommand{\jb}[1]{}
\renewcommand{\sj}[1]{}
\renewcommand{\yv}[1]{}
\fi

\renewcommand{\paragraph}[1]{\smallskip \noindent \textbf{#1.}}

\iffinal
\fi

%% file: macros.tex
\newcommand{\limp}{\rightarrow}

\newcommand{\Land}{\bigwedge}

\newcommand{\cT}{\mathcal{T}}
\newcommand{\Bad}{\mathit{Bad}}
\newcommand{\Init}{\mathit{Init}}
\newcommand{\Tr}{\mathit{Tr}}
\newcommand{\Inv}{\mathit{Inv}}
\newcommand{\I}{\mathit{I}}

\newcommand{\Reach}{\textsc{Reach}}

\newcommand{\sBad}{\hat{\mathit{Bad}}}
\newcommand{\sInit}{\hat{\mathit{Init}}}
\newcommand{\sTr}{\hat{\mathit{Tr}}}

\newcommand{\spectre}{\textsc{Spectre}\xspace}
\newcommand{\ttspec}{\texttt{spec}\xspace}

\newcommand{\true}{\mathit{true}}

\newcommand{\unchanged}[1]{\mathsf{id}(#1)}
\newcommand{\spec}{\ensuremath{\mathsf{spec}}\xspace}

\newcommand{\fence}{\ensuremath{\mathsf{fence}}\xspace}

\newcommand{\Pdr}{\textsc{Pdr}\xspace}

\newcommand{\lfalse}{\textsf{false}\xspace}
\newcommand{\ltrue}{\textsf{true}\xspace}

\newcommand\pc{\mathsf{pc}}
\newcommand\bbN{\mathbb{N}}

\makeatletter
\providecommand*{\cupdot}{%
  \mathbin{%
    \mathpalette\@cupdot{}%
  }%
}
\newcommand*{\@cupdot}[2]{%
  \ooalign{%
    $\m@th#1\cup$\cr
    \hidewidth$\m@th#1\cdot$\hidewidth
  }%
}
\makeatother

\newcommand{\Queue}{\mathcal{Q}}

\newcommand{\Spacer}{\textsc{Spacer}\xspace}
\newcommand{\spacer}{\Spacer}
\newcommand{\PDR}{\textsc{Pdr}\xspace}

\newcommand{\seahorn}{\textsc{SeaHorn}\xspace}
\newcommand{\FenceSynth}{\textsc{CureSpec}\xspace}

\newcommand{\FVars}{\mathcal{F}}

\newcommand{\Leak}{\textbf{SpecLeak}\xspace}
\newcommand{\AddFence}{\textbf{AddFence}\xspace}
\newcommand{\Mem}{\texttt{VInst}\xspace}

%% file: intro.tex
\section{Introduction}

Speculative execution is an indispensable performance optimization of modern processors:
by predicting how branching (and other) conditions will evaluate and speculatively continuing computation,
it can avoid pipeline stalls when data from other computations is still missing.
When this data arrives, in case of a correct guess the results of the computations can be committed. Otherwise they have to be discarded, and the correct results computed. 
However, even when the results are not committed to registers that are available at the software level, speculative computations may leave traces in the microarchitecture that can leak through \emph{side channels}, as demonstrated by the family of \spectre attacks~\cite{Spectre,Smotherspectre,EvaluationTransient}.
For example, the cache is usually not cleaned up after a misspeculation, enabling \emph{timing attacks} that can discover data used during speculation.

Since the discovery of \spectre, several countermeasures have been developed~\cite{DBLP:conf/dac/KhasawnehKSEPA19,DBLP:conf/micro/KirianskyLADE18,Spectre,DBLP:conf/micro/YanCS0FT19}. 
As neither speculation nor side channels can be removed from current hardware 
without sacrificing significant amounts of computing power, the problem is usually dealt with at 
the software level. Mitigations for \spectre usually prevent information leaks during speculative execution
by prohibiting ``problematic'' instructions from being executed speculatively. 

Most of the existing mitigations prevent some \spectre attacks, but are known to be incomplete~\cite{EvaluationTransient}, 
i.e., the modified code may remain vulnerable, in circumstances that may or may not be known.
In addition, there have been approaches that use formal methods to 
obtain code that is \emph{guaranteed} to be resilient against clearly defined types of \spectre attacks, and  
a formal notion of \emph{speculative constant-time security} has been proposed~\cite{DBLP:conf/pldi/CauligiDGTSRB20} that can give guarantees against timing attacks under speculation.
However, these methods all have certain shortcomings: 
either they require manual modification of the code if potential leaks are found~\cite{Spectector,DBLP:conf/csfw/CheangRSS19,DBLP:conf/pldi/CauligiDGTSRB20},
they do not precisely state the security guarantees of automatically hardened code~\cite{oo7}, or they are based on a security type system, which are known to be rather difficult to extend to different assumptions (e.g., attacker models) or guarantees~\cite{Blade}.

\paragraph{Motivating example}
Consider the three programs in \cref{fig:spec_fence_semantics}. 
The original program $P$ is shown in \subref{fig:original_program}.
It accesses a 
public array \texttt{a} at position \texttt{i} after checking that the access to \texttt{a} is in bounds.
In speculative execution, this bounds-check can be ignored, which enables to read
arbitrary program memory and store it (albeit temporarily) into \texttt{k}.
The information leak appears when \texttt{k} is used in another memory access to public array \texttt{b}, making \texttt{k} observable to an attacker through a cache-based timing attack.

Our repair approach is based on a transformation of $P$ with the following goals:
\begin{inparaenum}[(i)]
\item capture computations that are executed speculatively;
\item identify possible information leaks under speculation; and
\item enable the prevention of speculation using fences\footnote{
Other methods, e.g. \emph{speculative load hardening} (SLH, as appears in \url{https://tinyurl.com/3nybax4u}), can
also be used as a mitigation in our repair algorithm.}. 
\end{inparaenum}

	\vspace{-1em}
\begin{figure}
    \begin{subfigure}[b]{0.23\textwidth}
\begin{lstlisting}
if (i < size_a) {@\label{line:branch-orig}@
  k = a[i] * 512;@\label{line:fence-orig}@
  tmp = b[k];@\label{line:nested-read}@
}
\end{lstlisting}
    \caption{Original program}
    \label{fig:original_program}
    \end{subfigure}
    \hfill
    \begin{subfigure}[b]{0.37\textwidth}
\begin{lstlisting}
bool spec = false;
if (*) {@\label{line:branch}@
  spec = spec | !(i < size_a);@\label{line:spec}@
  k = a[i] * 512;
  assert(!spec);@\label{line:not-spec}@
  tmp = b[k];
}
\end{lstlisting}
    \caption{Speculative execution semantics}
    \label{fig:spec_program}
    \end{subfigure}
    \begin{subfigure}[b]{0.36\textwidth}
\begin{lstlisting}
bool spec = false;
bool fence2 = true;
bool fence3 = false;
if (*) {
  spec = spec | !(i < size_a);
  assume(!(fence2 && spec));@\label{line:fence2}@
  k = a[i] * 512;
  assume(!(fence3 && spec));@\label{line:fence3}@
  assert(!spec);
  tmp = b[k];
}
\end{lstlisting}
    \caption{Speculative execution and fence semantics}
    \label{fig:spec_fence_program}
    \end{subfigure}
\caption{Example program and its version with speculative and fence semantics}
\label{fig:spec_fence_semantics}
\end{figure}
	\vspace{-1em}

\cref{fig:spec_program} presents a modification $P_s$ of $P$, demonstrating the first point. 
We assume that speculative executions can start at 
conditional statements, i.e., the processor may ignore the condition and take the wrong branch.
Therefore, we replace the branching condition by a non-deterministic operator $*$ that can return either \ltrue or \lfalse (\cref{line:branch}).
Moreover, an auxiliary variable \ttspec is added in order to identify 
whether an execution of $P_s$ corresponds to a speculative or a non-speculative execution of $P$. 
Namely, $\ttspec=\ltrue$ at some point in an execution of $P_s$ iff 
the corresponding execution of $P$ is possible \emph{only} under speculation.
In \cref{fig:spec_program}, \ttspec is assigned %
\ltrue in \cref{line:spec} if the negation of the branching condition holds.

To detect information leaks under speculative execution, we assume that there is a set $\Mem$ of 
memory accesses that should not be performed under speculation, 
and for such memory accesses we add an assertion $\ttspec=\lfalse$.
Assuming that the nested array read in \cref{line:nested-read} of $P$ (\cref{fig:original_program}) is in $\Mem$, the transformed program $P_s$ (\cref{fig:spec_program}) contains such an 
assertion at \cref{line:not-spec}.

To enable prevention of speculation, in \cref{fig:spec_fence_program} we add auxiliary variables 
$\texttt{fence} i$ for every line $i$ in $P$ with an instruction, and initialize them with truth values that 
determines whether speculation should be stopped before reaching line $i$ of $P$.
We model the fact that fences stop speculation by adding assumptions that at line $i$ of $P$ we 
cannot have $\texttt{fence} i=\ltrue$ and $\ttspec=\ltrue$ simultaneously (\cref{line:fence2} and \cref{line:fence3} of \cref{fig:spec_fence_program}).
In this example, speculation can only start in \cref{line:branch-orig} of $P$ and $\texttt{fence}2=\ltrue$ stops speculation before \cref{line:fence-orig} of $P$, implying that the vulnerable instruction in \cref{line:nested-read} of $P$ is not reachable under speculation.

Note that in this example we assume that fence variables have fixed truth values, reducing the problem to a safety verification problem.
In our repair algorithm, we start with a program where all variables $\texttt{fence}i$ are initialized to $\lfalse$,
and allow the algorithm to manipulate initial values of the $\texttt{fence}i$ in order to find a version of the program that is secure against a given type of \spectre attacks (determined by our choice of vulnerable memory accesses $\Mem$).
Upon termination, our algorithm returns a list of instructions such that adding fence instructions in these positions (of the LLVM code) will make $P$ secure.

To formalize the ideas presented on this example, in \cref{sec:prelim} we will introduce a formal model for the standard semantics of a program $P$, then introduce a semantics that includes speculative executions in \cref{sec:spec-semantics}, and finally present our automatic repair approach in \cref{sec:repair}.

\paragraph{Our contribution}
In this paper, we present \FenceSynth. \FenceSynth is the first model checking based framework 
for \emph{automatic repair} of programs with respect to information leaks
that are due to speculative execution. 
Applying \FenceSynth to a given program results in a program with a \emph{certified} security guarantee.
\FenceSynth is \emph{parameterized by a threat model}, given as a set of instructions that may leak secret 
information to the attacker if executed under speculation. This makes our technique applicable to a wide 
range of speculative execution attacks, including different \spectre attacks that have been identified in the literature~\cite{EvaluationTransient}.
For a given threat model, \FenceSynth is able to either
prove that the program is secure, or \emph{detect} potential side-channel vulnerabilities.
In case vulnerabilities are detected, \FenceSynth \emph{automatically inserts mitigations} that remove these vulnerabilities, and proves
that the modified code is \emph{secure}. Since \FenceSynth is based on model checking, 
it provides a \emph{certificate} for the security guarantee in the form of an inductive invariant.

\FenceSynth is a framework with two main parts: 
\begin{inparaenum}[(i)]
\item a reduction of the problem of finding information leaks that are due to speculation
to a safety verification problem; and \label{enum:specleak}
\item a model checking based algorithm for detection and repair of possible vulnerabilities. \label{enum:repair}
\end{inparaenum}
For (i), we build on previous results that introduced formal modeling of speculative execution semantics~\cite{DBLP:conf/atva/BloemJV19}, and extend this formal model to enable not only the detection of possible leaks, but also their automatic repair.
For (ii), we extend the well-known IC3/\PDR approach~\cite{DBLP:conf/vmcai/Bradley11,DBLP:conf/sat/HoderB12,DBLP:journals/fmsd/KomuravelliGC16} %
to a repair algorithm for our problem.
When the underlying model checking algorithm 
discovers a possible leak, \FenceSynth modifies the program to eliminate this vulnerability. 
Then it resumes the verification process, and eliminates further vulnerabilities until the program is secure. 
An important feature of our technique is that the modified code is not verified from scratch, but the model checking algorithm maintains
its state and re-uses the information obtained thus far.
Finally, when \FenceSynth proves safety of the %
(possibly repaired) program, the underlying \PDR-algorithm produces 
an inductive invariant, which is a formal certificate 
of the desired security property.%

We implemented \FenceSynth in \seahorn~\cite{DBLP:conf/cav/GurfinkelKKN15}, a verification framework for C programs, and evaluated it on the ``standard'' test cases for \spectre vulnerabilities, as well as different parts of OpenSSL, demonstrating its practicality.

To summarize, in this paper we provide the first method that \emph{combines} the following aspects:
\begin{enumerate}
    \item formal verification of programs with respect to information leaks under speculation, parameterized by a threat model; 
    \item automatic repair by insertion of mitigations that stop speculative execution;
	  \item formal guarantees for the repaired code in form of inductive invariants; and 
    \item an implementation that scales to practical encryption algorithms.
\end{enumerate}

%% file: prelims.tex
\section{Preliminaries}
\label{sec:prelim}

\subsection{Model Checking Programs}
\label{sec:model-checker}
We consider first-order logic modulo a theory $\cT$ and denote it by $FOL(\cT)$.
We adopt the standard notation and terminology, where $FOL(\cT)$ is defined
over a signature $\Sigma$ of constant, predicate and function symbols,
some of which may be interpreted by $\cT$. In this paper $\cT$ is the theory
of Linear Integer Arithmetic and Arrays (LIA).
We use $\ltrue$ and $\lfalse$ to denote the constant truth values.

\paragraph{Transition Systems}
We define transition systems as formal models of programs.
Let $X$ be a set of variables, used to represent program variables.
A \emph{state} is a valuation of $X$.
For a state $\sigma$ and 
 $a\in X$, we denote by $\sigma[a]$ the value of $a$ in $\sigma$. 
We write $\theta(X)$ to represent a formula over $X$ in $FOL(\cT)$. 
$\theta(X)$ is called a \emph{state formula} and 
represents a set of states.

A \emph{transition system} is a tuple $M=\langle X, \Init(X), 
\Tr(X,X' )\rangle$ where $\Init(X)$ and $\Tr(X,X')$ are quantifier-free 
formulas in $FOL(\cT)$. $\Init$
represents the initial states of the system and $\Tr$ represents the
(total) transition relation. We write $\Tr(X,X')$ to denote that $\Tr$ is
defined over variables $X \cup X'$, where $X$ is
used to represent the pre-state of a transition, and
$X' = \{a' \mid a \in X\}$ is used to represent the post-state.
A \emph{path} in a transition system is a sequence of states 
$\pi := \sigma_0, \sigma_1, \ldots$, such that for 
$i \geq 1$: $(\sigma_{i-1}, \sigma_i')\models\Tr$. We also consider the case where
a path is a finite sequence of states such that every two subsequent states have a transition.
We use $\pi[i]$ to refer to the $i$-th state of $\pi$, namely $\sigma_i$.
We use $\pi^{[0..n]}$
to refer to the prefix $\sigma_0, \sigma_1, \ldots, \sigma_n$, and $\pi^n$ to the
suffix $\sigma_n,\ldots$ of $\pi$.
A path (or a prefix of a path) is called an \emph{execution} of $M$ when $\sigma_0\models\Init$.
Given two paths $\pi_1 = \sigma_{0},\sigma_1\ldots$ and
$\pi_2 = \sigma'_{0},\sigma'_1\ldots$, then $\pi = \pi_1^{[0..n]}\pi_2$
is a path if $(\sigma_{n},\sigma'_{0})\models\Tr$.

A formula $\varphi(X,X')$ such that for every valuation $\I$ of $X$ there is 
exactly one valuation $\I'$ of $X'$ such that $(\I, \I')\models \varphi(X,X')$ is called a 
\emph{state update function}. For $Y \subseteq X$, we denote by $\unchanged{Y,Y'}$ the state 
update function $\bigwedge_{a \in Y} a'=a$. While $\unchanged{Y,Y'}$ is a formula over $Y\cup Y'$,
for readability we use $\unchanged{Y}$.

\paragraph{Safety Verification} 
A \emph{safety problem} is a tuple
$\langle M, \Bad(X) \rangle$, where
$M=\langle X, \Init, \Tr \rangle$ is a transition system and $\Bad$ is a
quantifier-free formula in $FOL(\cT)$ representing a set of
bad states.
A safety problem has a \emph{counterexample of length $n$} if
there exists an execution $\pi := \sigma_0,\ldots,\sigma_n$ with $\sigma_n\models\Bad$.
The safety problem is \emph{SAFE} if it has no counterexample, of any length.
It is \emph{UNSAFE} otherwise.

A \emph{safe inductive invariant} is a formula $\Inv(X)$ such that
\begin{inparaenum}[(i)]
\item $\Init(X) \limp \Inv(X)$,
\item $\Inv(X) \wedge \Tr(X,X') \limp \Inv(X')$, and
\item $\Inv(X)\limp \neg \Bad(X)$.
\end{inparaenum}
If such a safe inductive invariant exists, then the safety problem is SAFE.

In this work we use \Spacer~\cite{DBLP:journals/fmsd/KomuravelliGC16} 
as a solver for a given safety problem. \Spacer is based on the 
Property Directed Reachability (\PDR) algorithm~\cite{DBLP:conf/vmcai/Bradley11,DBLP:conf/sat/HoderB12}.
Algorithm~\ref{alg:spacer} presents \spacer as a set of rules, following the presentation style
of~\cite{DBLP:conf/sat/HoderB12}.
We only give a brief overview of \PDR and \Spacer and highlight the
details needed later in the paper for \FenceSynth.
Given a safety problem, \Spacer tries to construct an inductive invariant, or find a counterexample.
In order to construct an inductive invariant, \Spacer maintains a sequence of formulas 
$F_0, F_1, \ldots, F_N$, with the following properties:
\begin{inparaenum}[(i)]
\item $F_0\limp\Init$;
\item $\forall 0\leq j < N\cdot F_j\limp F_{j+1}$;
\item $\forall 0\leq j < N\cdot F_j(X)\land\Tr(X,X')\limp F_{j+1}(X')$; and
\item $\forall 0\leq j < N\cdot F_j\limp \neg\Bad$.
\end{inparaenum}

$F_j$ is an over-approximation of the states reachable in $j$ steps
or less. Additionally, \Spacer maintains a set \Reach\
of states that are known to be reachable. \Reach\ is an under-approximation
of the reachable states.

\PDR performs a backward traversal of the states space. The
traversal is performed starting from 
states that violate $\Bad$ and constructing a suffix of a counterexample
backwards, trying to either show that a state that can reach $\Bad$
is reachable (\textbf{Candidate}, \textbf{Predecessor} and \textbf{Cex} rules)
or prove that such states are unreachable (\textbf{NewLemma} and \textbf{Push} rules).
During this process the trace of over-approaximations $F_0, F_1,\ldots, F_N$
is constructed and reachable states are discovered. We later show how
these are used in \FenceSynth. For more details about \PDR and \Spacer
the interested reader is referred to~\cite{DBLP:conf/vmcai/Bradley11,DBLP:conf/sat/HoderB12,DBLP:journals/fmsd/KomuravelliGC16}.

\begin{algorithm}[t!]\small
  \KwIn{A safety problem
    $\langle X, \Init(X), \Tr(X,X'), \Bad(X) \rangle$.}

  \KwSty{Assumptions}: $\Init$, $\Tr$ and $\Bad$ are quantifier free.

  \KwData{A queue $\Queue$ of potential counterexamples, where $c \in \Queue$ is a pair
      $\langle m, j \rangle$, $m$ is a cube over state variables,
      $j \in \bbN$.
      A level $N$.  A sequence
      $F_0, F_1, \ldots$. An invariant $F_\infty$. A set of reachable
      states $\Reach$.}

  \KwOut{(\emph{SAFE}, $F_\infty$), where $F_\infty$ is a safe inductive invariant, or \emph{Cex} }

  \KwSty{Initially:}
    $\Queue = \emptyset$, $N=0$, $F_0 = \Init$,
    $\forall j \geq 1 \cdot F_j=\true$, $F_\infty = \true$.\\
    \KwSty{Require:} $\Init \to \neg \Bad$\\
    \Repeat{$\infty$} {
      \begin{description}
      \item[Safe] If $F_{\infty}\to\neg\Bad$
        \Return (\emph{SAFE}, $F_\infty$).
      \item[Cex] If $\langle m, j \rangle \in \Queue$ and ,
        $m\cap ( \Reach) \ne \emptyset$ 
        \Return \emph{Cex}.
      \item[Unfold] If $F_N \rightarrow \neg \Bad$, then set
        $N \gets N + 1$.
      \item[Candidate]
        If for some $m$, $m \to F_N \land \Bad$,
        then add $\langle m, N\rangle$ to $\Queue$.
      \item[Predecessor] If $\langle m, j+1\rangle \in \Queue$ and there
        are $m_0$ and $m_1$ s.t.\\ $m_1 \to m$, $m_0 \land m'_1$ is
        satisfiable, and $m_0 \land m_1' \to F_{j} \land \Tr
        \land m'$, then\\  add $\langle m_0, j\rangle$ to
        $\Queue$.
      \item[NewLemma] For $0 \leq j < N$: given 
        $\langle m, j+1\rangle\in\Queue$ and a clause $\varphi$ s.t. $\varphi\to \neg m$,\\
				if $(\Reach) \to \varphi$, and
        $\varphi \land F_j \land \Tr \to \varphi'$, then add $\varphi$
        to $F_{k}$, for $k \leq j + 1$.
      \item[ReQueue] If $\langle m, j \rangle \in \Queue$, and
        $F_{j-1} \land \Tr \land m'$ is unsatisfiable, then \\ add
        $\langle m, j+1\rangle$ to $\Queue$.
      \item[Push] For $1 \leq j$ and a clause
        $(\varphi \lor\psi) \in F_j\setminus F_{j+1}$,\\ 
        if $(\Reach) \to \varphi$ and $\varphi \land F_j \land \Tr \to \varphi'$, then\\ 
        add $\varphi$ to $F_{k}$, for each $k \leq j + 1$.
     \item[MaxIndSubset] If there is $j > N$ s.t. $F_{j+1} \subseteq F_j$,
        then \\
        $F_\infty \gets  F_i$, and $\forall k \geq j \cdot F_j \gets F_\infty$.
      \item[Successor] If $\langle m,j+1\rangle \in \Queue$ and exist
        $m_0$, $m_1$ s.t.  \\
        $m_0\land m'_1$ are satisfiable and
        $m_0\land m'_1 \to (\Reach) \land \Tr \land m'$, then \\
        add $m_1$ to $\Reach$.
      \item[ResetQ] $\Queue \gets \emptyset$.
      \item[ResetReach] $\Reach \gets \Init$.
    \end{description}
		\vspace{-1em}%
  }

  \caption{The rules of the \Spacer procedure.}
  \label{alg:spacer}
\end{algorithm}

\subsection{Standard Program Semantics}\label{sec:program_semantics}
We assume a program $P$ is represented in
a low-level language (e.g. LLVM bit-code) with standard semantics,
and includes standard low-level instructions such as unary and binary
operations, conditional and unconditional branches, load and store for 
accessing memory. In addition, it includes the instructions $\textsf{assume}(b)$
and $\textsf{assert}(b)$ used for safety verification.

Let $i\in \bbN$ be a line in the program $P$ to be encoded. For simplicity 
we refer to the instruction at line $i$ as $i$, and write $i\in P$.
We assume there is a special program variable $\pc\in X$, called the \emph{program counter},
defined over the domain $\bbN\cup\{\bot\}$.
Let $i\in P$ be an instruction to be encoded.
If $i$ is a conditional branch instruction, it is encoded by a 
\emph{conditional state update function} of the form 
\[\tau_i(X,X') := \pc = i \rightarrow (\mathit{cond}_i(X)\ ?\ \pc'=\vartheta_i(X) : \pc'=\varepsilon_i(X))\land\unchanged{X\setminus\{\pc\}},\]
where $\mathit{cond}_i(X)$ is the condition represented by a state formula
and $\pc$ is updated to $\vartheta_i(X)$ when the condition holds,
and otherwise it is updated to $\varepsilon_i(X)$.

All other instructions are encoded by an \emph{unconditional state update function} of the form
$\tau_i(X,X') := \pc = i \rightarrow  \varphi_i(X,X')$,
where $\varphi_i(X,X')$ is a state update function
\footnote{Note that $\pc$ is updated also by $\varphi_i$.
We therefore assume that unconditional state update functions
accompany every instruction that is not a conditional branch.}.
Instructions are either conditional or unconditional
depending on their corresponding state update function. We denote by $C\subseteq P$ the set of conditional
instructions, i.e., if $i\in C$, then $\tau_i$ is a conditional state update function.

The semantics for the verification instruction $\textsf{assume}(\mathit{cond}_i(X))$
is captured by a state update function, which is encoded by 
\begin{align*}
  \tau_i(X,X'):= \pc = i \rightarrow ((\mathit{cond}_i(X)\ ?\ \pc' = \vartheta_i(X) \ : \pc'=\pc)\wedge \unchanged{X \setminus \{\pc\}}).
\end{align*}
This encoding requires that at line $i$ 
the condition $\mathit{cond}_i(X)$ holds. If it does not hold, then the transition relation is stuck 
in an infinite loop and the program does not progress. Similarly, $\textsf{assert}(\mathit{cond}(X))$
is captured by a state update function, which is encoded by
\begin{align*}
  \tau_i(X,X'):= \pc = i \rightarrow ((\mathit{cond}_i(X)\ ?\ \pc' = \vartheta_i(X) \ : \pc'=\bot)\wedge \unchanged{X \setminus \{\pc\}})
\end{align*}

For $\textbf{assert}$, if the condition $\mathit{cond}_i(X)$ holds, the program continues.
Otherwise, $\pc$ is set to $\bot$. This special case allows us to create a safety verification
problem by defining the bad states to be those where $\pc = \bot$. To ensure the resulting transition
relation is total, we add a state update function that makes sure that if a state where $\pc=\bot$
is ever reached, this state is stuttering:
\begin{align*}
  \tau_{\bot}(X,X') := pc=\bot\rightarrow \unchanged{X}
\end{align*}

To conclude, given a program $P$, we obtain a symbolic representation 
of the transition relation by conjoining 
the formulas for all lines of the program including $\tau_{\bot}$, i.e., $\Tr(X,X') := (\Land_i \tau_i(X,X'))\land\tau_{\bot}(X,X')$.
The corresponding safety problem is then defined by the resulting transition system $M=\langle X,\Init,\Tr\rangle$ and a set of bad states, given as $\Bad(X) := \pc =\bot$.

\begin{remark}\label{cor:tr_inst}
Let $P$ be a program, $\langle M, \Bad\rangle$ the corresponding safety problem, and
$\pi=\sigma_0,\sigma_1,\ldots,\sigma_n$ an execution of $M$.
Then, for every $1\leq j\leq n$ such that $(\sigma_{j-1},\sigma_{j}')\models\Tr$
there exists $i\in P$ such that $(\sigma_{j-1},\sigma_{j}')\models\tau_i$.
We denote this %
as $i\in \pi$. Moreover, if $\sigma_n\models\Bad$ then
there exists an assertion violation in $P$.
\end{remark}

%% file: speculative-semantics.tex
\section{Modeling Speculative Execution Semantics}
\label{sec:spec-semantics}
When analyzing the \emph{functionality} of a program, speculative execution can be ignored since 
the results of a computation that is based on \emph{misspeculation} do not alter the program's state: 
if the condition of a branch turns out to be wrong after speculative execution, any of its results (visible in the microarchitecture, but not on the program level) are discarded and computation backtracks to the correct branch.
However, the data used in such a computation may still leak through side channels.
Therefore, when analyzing \emph{information leaks through side channels}, 
the formal model must include speculative execution semantics and take into account possible observations 
based on misspeculation.

In this section, we first discuss the notion of security and threat models we consider.
Later we give 
a formal definition of the speculative execution model in Section~\ref{sec:fm_spec}.

\subsection{Threat Models for Speculative Information Leaks}

Verifying secure information flow deals with proving that confidential data does not flow 
to public outputs during the execution of a system~\cite{DBLP:journals/cacm/DenningD77}. 
Another way of describing secure information flow is by specifying that confidential data
is not \emph{observable} by an attacker.

A well-known class of attacks that can cause information leaks are \emph{timing-attacks}.
These attacks use observations about the run-time of a system in order to infer
secret data. More precisely, in order for a program to be secure against timing-attacks
any two executions of a system where the public
inputs and outputs are equivalent for both executions, should be indistinguishable
w.r.t. some measure of time (e.g. time to execute a program, latency in
memory access, etc.).

Most variants of the \spectre attack fall within this class of timing-attacks.
These attacks use side-effects of cache collisions caused by code that
executes speculatively in order to leak secret information. It is important
to note that speculative execution does not change the architectural
state of the CPU and only has side-effects (e.g. modifications to the cache).
Consequently, if we consider a program to be secure with respect to standard 
execution semantics, given two different 
executions of that program that start from the same public inputs, where one
execution uses speculative execution and the other does not, both produce
the same public outputs. However, due to the side-effects caused by speculative
execution, these two executions may still be distinguishable.

We therefore use the following definition of Constant-Time Security (cp.~\cite{DBLP:conf/uss/AlmeidaBBDE16}):
\vspace{-1em}
\begin{definition}\label{def:security}
Let $M$ be a transition system of a program $P$, let $H\subseteq X$ be a set of
high-security variables and $L := X \setminus H$ a set of low-security
variables. 
$M$ is \emph{Constant-Time Secure} if for
any two executions $\pi_1$ and $\pi_2$ of $M$ with 
$\forall x\in L,i\in\bbN\cdot\pi_1[i][x] =\pi_2[i][x]$ are also indistinguishable.
\end{definition}

Since the focus of this paper is timing-attacks that can incur due to \spectre,
we make the standard assumption that the attacker can \emph{distinguish} executions if they differ in the values of $\pc$, i.e., their control-flow, or in the location of certain memory accesses.
We collect these \emph{vulnerable memory instructions} in a set $\Mem \subseteq P$.

In practice, identifying these instructions depends on our assumptions on the attacker capabilities in the given setting%
\footnote{Compared to the work in~\cite{DBLP:conf/atva/BloemJV19}, our threat model 
is parameterized in the set of vulnerable instructions and strictly generalizes the fixed threat model used there.}.
We consider two variants of \textsc{Spectre-} \textsc{PHT} vulnerabilities that are based on the classification by Canella et al.~\cite{EvaluationTransient}:
\begin{itemize}
\item we consider a \emph{strong} \spectre model, where the attacker is very powerful and can observe the value of any array access, i.e., any memory access (e.g. array access $a[i]$) is in $\Mem$.
  This is motivated by the use of side-channels other than cache timing attacks (and by advanced mistraining strategies for the speculation unit~\cite{EvaluationTransient}).
\item in the \emph{classical} \spectre model, an attacker uses the cache content for timing attacks, i.e., 
if the location $i$ of an array access $a[i]$ is controlled by the attacker (e.g., the attacker can directly provide it as an input, or the inputs of the attacker have some influence on $i$), then a nested array access $b[a[i]]$ can be used to reveal the content of $a[i]$.
Therefore, here we consider $\Mem$ to contain all array accesses that amount, directly or indirectly, to a nested array access at an attacker-controlled position.
Detecting such array accesses can be done statically, using information flow analysis techniques like taint tracking or self-composition.
\end{itemize}

\vspace{-1em}
\subsection{Formal Model of Speculative Execution Semantics}\label{sec:fm_spec}
Our goal is to check for information leaks \emph{under speculative execution}.
To this end, we \emph{assume} that the given program $P$ is constant-time secure in the absence of speculation.
We will show that this allows us to reduce the problem of detecting speculative information leaks to the standard safety property of checking 
whether instructions in the given set $\Mem$ are reachable under speculation in $P$.
First, let us formalize the speculative semantics of $P$.

Let $M=\langle X, \Init, \Tr\rangle$ be the standard transition system (\cref{sec:model-checker}) of the program $P$
and $\Mem$ the set of vulnerable instructions.
Define the set of Boolean auxiliary variables
used to model fences as $\FVars := \{\fence_i \mid i \in P\}$, and 
let $B\subseteq P$ be the set of fenced instructions, i.e., with $\fence_i = \ltrue$.
Then, the transition system that includes speculative execution
semantics is defined as 
$\hat{M} := \langle \hat{X},\hat{\Init},\hat{\Tr}\rangle$ where
$\hat{X} := X\cup \{ \spec \} \cup \FVars$, where $\spec$ is of sort $\mathbb{N}_0$. 
The initial states of $\hat{M}$ are defined as
\begin{align*}
\hat{\Init}(\hat{X}):= & \Init\land(\spec=0) \land\Land_{i \in B}\fence_{i}\land\Land_{i \in P\setminus B} \neg\fence_{i} 
\end{align*}
where \ttspec is initialized to $0$, and
auxiliary variables in $\FVars$ are initialized to $\ltrue$ if the corresponding
instruction is fenced, and otherwise to $\lfalse$.

To define $\hat{\Tr}$, recall that speculation starts if the wrong branch is taken for some $i\in C$.
At the first such position, $\spec$ becomes positive and remains positive 
for the rest of the execution. In order to formally model this behavior
we modify the state update functions in the following manner.

\paragraph{Conditional Instructions} The state update function for each 
conditional instruction $i\in C$ 
(as it appears in \cref{sec:program_semantics}) is defined as follows:
\begin{align*}
  \hat{\tau}_i(\hat{X}, \hat{X}') & :=\pc = i \rightarrow \Big( \fence_i \land \spec > 0\ ?\ \unchanged{\hat{X}} : \unchanged{\hat{X}\setminus\{\spec,\pc\}} \,\land \\
  & \Big[ ((\spec' = ((\neg\mathit{cond}(X)\lor \spec > 0) \ ? \ \spec + 1 : 0))\land\pc'=\hat{\vartheta}(X)) \,\lor \\
  &  ((\spec' = ((\mathit{cond}(X)\lor\spec>0) \ ? \ \spec + 1: 0))\land\pc'=\hat{\varepsilon}(X)) \Big] \Big)
\end{align*}

Note that $\hat{\tau}_i$ is stuck in an infinite loop in case $\spec>0$ and $\fence_i$ is set to $\ltrue$.
Otherwise, if the respective branch condition does not hold, the value 
of $\spec$ has to become positive.
If $spec$ is already positive, it remains positive and is incremented.
Overall, $\spec$ can only be positive in a given execution iff at least one branch 
condition of the execution is not met. 
Careful inspection of the state update
function reveals the updates to $\pc$ are different (e.g. $\hat{\vartheta}(X)$
instead of $\vartheta(X)$). We address this later in the section ($\clubsuit$).

\paragraph{Unconditional Instructions} For unconditional instructions, 
we must take into account the new
auxiliary variables, as well as the fence assumptions that prevent speculation.
Given an unconditional instruction $\tau_i(X,X') := pc=i\limp\varphi(X,X')$
we define $\hat{\tau}_i$ in the following manner:
\begin{align*}
  \hat{\tau}_i(\hat{X},\hat{X}') := pc = i \limp & \left.(\fence_i\land \spec>0\ ?\ \unchanged{\hat{X}} : \right.\\
  & \quad\ \hat{\varphi}(X,X') \land \unchanged{\FVars} \land \spec'= (\spec > 0\ ?\ \spec+1 : 0))
\end{align*}
Again, $\hat{\tau}_i$ enters an infinite loop in case $\spec>0$ and $\fence_i$ is set to $\ltrue$.

\paragraph{Speculation Bound}
In order to allow a realistic modeling of speculative execution,
we consider a model that only allows a bounded speculation window. This modeling
comes from the fact that, in any given microarchitecture, the Reorder Buffer (ROB) used to allow out-of-order
execution is limited in the number of instructions it can occupy.%

Therefore, we assume that for a given micro-architecture
there exists a parameter $\Bbbk$ that is a bound on speculative executions.
In order to enforce this bound, we constrain the transition relation $\hat{\Tr}$
such that if $\spec$ ever reaches $\Bbbk$, $\hat{M}$ is stuck in an infinite loop
(this can be viewed as a global assumption).

This results in the following two cases for the transition relation:
\begin{align*}
\hat{\Tr}_{<\Bbbk}(\hat{X}, \hat{X}') := &  (\spec < \Bbbk)\land 
  \Land_{i\in P} \hat{\tau}_i(\hat{X}, \hat{X}')\\
  \hat{\Tr}_{\geq\Bbbk}(\hat{X}, \hat{X}') := &  (\spec >=\Bbbk)\land \unchanged{\hat{X}}
\end{align*}

\paragraph{Speculative Constant-time Security as a Safety Property}
We want to ensure that $P$ does not leak information under speculation.
We assume that the attacker has control over speculation, which implies that $\spec \in L$.
Therefore, the condition of \cref{def:security} needs to hold for executions that make the same speculative choices (cp. \emph{speculative constant-time}~\cite{DBLP:conf/pldi/CauligiDGTSRB20,DBLP:conf/ndss/DanielBR21}).
Since we assume that the transition system $M$ (without speculation) is constant-time secure,
we therefore only have to compare speculative executions that agree on their control-flow.
Since we also assume that additional side-channel observations of the attacker are only possible when a vulnerable instruction in $\Mem$ is executed, our problem reduces to the standard safety property of checking whether instructions in $\Mem$ are \emph{reachable under speculation}.
This can be modeled by adding an assertion instruction $\mathbf{assert}(\spec==0)$ before every such instruction.

In the transition relation, this is reflected by defining
\begin{align*}
  \hat{\tau}_{a_i}(\hat{X},\hat{X}'):= \pc = a_i \rightarrow ((\spec=0)\ ?\ \pc' = i \ : \pc'=\bot)\wedge \unchanged{\hat{X} \setminus \{\pc\}})
\end{align*}
and letting
\[ \hat{\Tr}(\hat{X}, \hat{X}') := \left( \hat{\Tr}_{<\Bbbk}(\hat{X}, \hat{X}') \land \Land_{i\in\Mem}\hat{\tau}_{a_i}(\hat{X},\hat{X}')\land\tau_{\bot}(\hat{X},\hat{X}') \right) \lor \hat{\Tr}_{\geq\Bbbk}(\hat{X}, \hat{X}').\]

The addition of these assertions is the reason for the different $\pc$ updates in $\hat{\tau}$ ($\clubsuit$).
Namely, an instruction $j\in P$ that precedes an instruction $i\in\Mem$ now
needs to precede the corresponding assertion $a_i$.

Now, we can encode constant-time security with respect to
speculative execution as the safety problem $\langle \hat{M}, \hat{\Bad}(X) \rangle$, where
\[\hat{\Bad}(X) := \pc =\bot,\]
i.e., the definition of bad states remains unchanged.

\paragraph{Properties of the Speculative Execution Semantics}
We give some useful properties of the above semantics.
The proofs can be found in \cref{sec:appendix-proofs}.

\begin{restatable}{lemma}{simulationlemma}
  \label{lemma:simulation}
Let $P$ be a program, $M=\langle X, \Init, \Tr\rangle$ its transition system, and 
$\hat{M}$ the transition system including speculative execution semantics.
Then, there exists a simulation relation between $M$ and $\hat{M}$, denoted $M\leq_{\text{sim}}\hat{M}$.
\end{restatable}

\begin{restatable}{lemma}{gspeclemma}
  \label{lemma:gspec}
Let $\pi=\sigma_0,\sigma_1,\ldots,\sigma_k$ be an execution of $\hat{M}$ 
such that $\sigma_k\models\sBad$. Then the following conditions hold:
\begin{inparaenum}[(i)]
\item $\sigma_0\models\spec=0$
\item $\sigma_k\models\spec>0$
\item There exists a unique $0\leq j < k$ such that $\sigma_j\models\spec=0$ and
$\sigma_{j+1}\models\spec>0$.
\end{inparaenum}
\end{restatable}

\begin{restatable}{lemma}{soundnesslemma}
  \label{lemma:soundness}
Let $P$ be a program, $M=\langle X, \Init, \Tr\rangle$ its transition system, and $\Mem \subseteq P$.
Assume that $M$ is constant-time secure, where observations of the attacker are determined by $\Mem$.
Let $\hat{M} = \langle \hat{X},\hat{\Init},\hat{\Tr}\rangle$ be the transition system 
of $P$ that includes speculative execution semantics.
If $\hat{M}$ is SAFE wrt. $\hat{Bad}$, then $\hat{M}$ is constant-time secure, i.e., $P$ does not leak information under speculation.
\end{restatable}

%% file: repair.tex
\section{Automatic Repair under Speculative Execution}
\label{sec:repair}

In this section we describe \FenceSynth, an automatic, model checking based repair algorithm,
that fixes \spectre related information leaks under speculative execution.
All proofs can be found in \cref{sec:appendix-proofs}.
\FenceSynth receives a program $P$ and a set of vulnerable instructions $\Mem$.
If \FenceSynth detects that an instruction $i\in\Mem$
is executed under speculation, it repairs $P$ by analyzing 
the leaking execution and adding a fence instruction
that disables speculation in program locations that enabled the
leak. This process is iterative and continues
until \FenceSynth proves that $P$ is secure. 

Given a program $P$, its speculative transition system $\hat{M}$, and the 
safety problem $\hat{T}=\langle\hat{M},\hat{\Bad}\rangle$ (as 
defined in \cref{sec:spec-semantics}, w.r.t. a set of vulnerable instructions $\Mem$). 
Recall that 
\cref{alg:spacer} can determine if $\hat{T}$ is SAFE or UNSAFE. 
When \cref{alg:spacer} returns SAFE, then $\hat{M}$ is constant-time secure
with respect to $\Mem$ and speculative execution. 
Otherwise, a counterexample describing a speculative execution that may leak information 
is returned and \cref{alg:spacer} terminates.

\FenceSynth builds upon \cref{alg:spacer} and extends it such that if 
$\hat{T}$ is UNSAFE, instead of terminating, 
repair is applied. The repair process is iterative - when a leak is detected,
it is analyzed, a fence is added to mitigate the leak, and verification
is reapplied on the repaired program. However, \FenceSynth is
\emph{incremental} in the sense that it maintains the state of \cref{alg:spacer}
as much as possible and reuses it when verification re-executes.
To this end, \FenceSynth includes all rules of
\cref{alg:spacer} excluding the \textbf{Cex} rule, and including 
two additional rules as described in \cref{alg:spacer_repair}:
\Leak
and \AddFence. 

\paragraph{\Leak} This rule replaces the rule \textbf{Cex} 
from \cref{alg:spacer} and prevents the algorithm from terminating
when a counterexample is found. Instead, the leaking execution is stored
in $\pi$.

\paragraph{\AddFence} This rule is responsible for the repair. 
The trace $\pi$ is analyzed, and based on the misspeculation
that leads to a information leak, a fence that makes $\pi$ unfeasible is added
(by letting $\fence_k = \ltrue$ in $\hat{Init}$ for some $k$).
Note that the trace of $F_i$ and the invariant $F_\infty$ are unchanged which ensures incrementality of the overall algorithm.

The \textbf{Safe} rule is amended to additionally return the list of added fences.
\vspace{-1.5em}
\begin{algorithm}\small
  \KwIn{A safety problem
    $\langle \hat{M}, \sBad(\hat{X}) \rangle$ with $\hat{M} = \langle \hat{X},\sInit,\sTr\rangle$.}

  \KwSty{Assumptions}: $\sInit$, $\sTr$ and $\sBad$ are quantifier free.

  \KwData{A queue $\Queue$ of potential counterexamples, where $c \in \Queue$ is a pair
      $\langle m, j \rangle$, $m$ is a cube over state variables,
      $j \in \bbN$.
      A level $N$.  A trace
      $F_0, \ldots, F_N$. An invariant $F_\infty$. A set of reachable
      states $\Reach$.}

  \KwOut{A list of added fences $\mathcal{L}$ and a safe inductive invariant $F_\infty$}

  \KwSty{Initially:}
    $\Queue = \emptyset$, $N=0$, $F_0 = \sInit$,
    $\forall j \geq 1 \cdot F_j=\true$, $F_\infty = \true$,\\
    $\pi =\langle\rangle$, $\mathcal{L} = \emptyset$.\\
    \KwSty{Require:} $\sInit \to \neg \sBad$\\
    \Repeat{$\infty$}{%
      \begin{description}
      \item[\Leak]
        If $\langle m, j \rangle \in \Queue$ and
        $m\cap \Reach \ne \emptyset$.\\
        Let $\pi'=\langle m=\sigma_j, \sigma_{j+1},\ldots,\sigma_{N}\rangle$ be
        a path of $\hat{M}$  where $\sigma_N\models\sBad$.\\
        Then $\pi'$ can be extended into an execution
        $\pi = \pi^{*}\pi'$ of $\hat{M}$.
      \item[\AddFence]
        If $\pi\neq \langle\rangle$, choose $j\leq k\leq N$ where %
         $\sigma_k\models\spec>0$.
         Modify $\sInit$\\ s.\ t.\ $\sInit\models\fence_k$.
         $\mathcal{L}\gets\mathcal{L}\cup\{\fence_k\}$.
        $\Queue\gets\emptyset$, $\Reach\gets\sInit$, $\pi\gets\langle\rangle$.
    \end{description}
		\vspace{-1em}%
  }
  \caption{Repair algorithm \FenceSynth}
  \label{alg:spacer_repair}
\end{algorithm}
\iffinal\vspace{-3em}\fi

\subsection{Analyzing the Leaking Execution $\pi$}

When \cref{alg:spacer_repair} detects a potential leak in $\hat{M}$, 
\textbf{SpecLeak} stores this execution in $\pi$. In \AddFence, \cref{lemma:gspec} is used
to identify where speculation starts in $\pi$. Based on that, we define the following:
\begin{definition}
Let $\pi=\sigma_0,\ldots,\sigma_N,\ldots$ be an execution of $\hat{M}$ such that
$\sigma_N\models\sBad$, and $0<k\leq N$ such that $\sigma_{k-1}\models\spec=0$ and $\sigma_{k}\models\spec>0$.
Then 
$\pi^{[0..k-1]} = \sigma_0,\ldots,\sigma_{k-1}$ is called the \emph{non-speculating prefix} of $\pi$, and 
$\pi^{k} = \sigma_{k},\ldots$ the \emph{speculating suffix} of $\pi$.
We call $k$ the \emph{speculative split point} of $\pi$.
\end{definition}

By construction, we have $\spec=0$ in $\pi^{[0..k-1]}$. Therefore, letting $\fence_{i} = \ltrue$ for any instruction $i\in\pi^{[0..k-1]}$
has no effect on the transitions in $\pi^{[0..k-1]}$. 
However, since $\spec>0$ in $\pi^{k}$, letting $\fence_{i} = \ltrue$ for any $i\in\pi^{k}$
makes $\pi$ an unfeasible path. In order to formalize this intuition
we use the following definition and lemmas.

\begin{definition}
Let $P$ be a program and $\hat{M}=\langle \hat{X},\sInit,\sTr\rangle$ 
its speculative execution transition system. If for $i\in P$ it holds that 
$\sInit\models\neg\fence_i$, then adding a fence to $i$ results in a
new transition system
$\hat{M}_i = \langle \hat{X},\sInit_i,\sTr\rangle$
where $\sInit_i := \sInit[\neg\fence_i \leftarrow \fence_i]$.
\end{definition}

Here, $\sInit[\neg\fence_i \leftarrow \fence_i]$ denotes 
the substitution of $\neg\fence_i$ with $\fence_i$. After substitution
it holds that $\sInit_i\models\fence_i$, and initialization for all other 
variables is unchanged.
Thus, the initial value of $\fence_i$ is the only difference between 
$\hat{M}_i$ and $\hat{M}$.

\begin{restatable}{lemma}{repairlemma}
  \label{lemma:repair}
Let $\pi=\sigma_0,\ldots,\sigma_N$ be an execution of $\hat{M}$ such that
$\sigma_N\models\sBad$ and $k \leq N$ be its speculative split point. 
For any instruction $i \in P$ such that $i \in \pi^{k}$ (see \cref{cor:tr_inst}), $\sInit\models\neg\fence_i$. 
Moreover, $\pi$ is not an execution of $\hat{M}_i$ and for every execution 
$\hat{\pi} = \hat{\sigma}_0,\ldots,\hat{\sigma}_m$ 
of $\hat{M}_i$, $\hat{\pi}\pi^{k}$ is not an execution of $\hat{M}_i$.
\end{restatable}

By \cref{lemma:repair}, it is sufficient to let $\fence_i = \ltrue$ for 
any instruction $i\in\pi^{k}$ in order to make $\pi$ an unfeasible
path. Then, \FenceSynth can be resumed 
and search for a different leaking execution, or
prove that the added fences make $P$ secure.
This results in the following properties of our repair algorithm.

\vspace{-.5em}
\subsection{Properties of \FenceSynth (\cref{alg:spacer_repair})}
As noted earlier, \FenceSynth is \emph{parametrized} by the set
$\Mem$ of possibly vulnerable instructions. While in this paper
we focus on instructions that are vulnerable to \spectre,
\FenceSynth can detect other kinds of instructions that
are executed under speculation, and as a result, repair
such instances where execution under speculation may
lead to unwanted transient behavior.

\begin{restatable}{theorem}{correctnessthm}
  \label{thm:correctness}
Let $P$ be a program and $\hat{M}=\langle \hat{X},\sInit,\sTr\rangle$ 
its speculative execution transition system.
Then, on input $\langle \hat{M}, \sBad \rangle$, \FenceSynth
terminates and returns a list $\mathcal{L}$ such that
for $\hat{M}_s=\langle \hat{X},\sInit_s,\sTr\rangle$, 
where $\forall i\in\mathcal{L}\cdot\sInit_s\models\fence_i$, it holds that
$\langle \hat{M}_s,\sBad\rangle$ is SAFE (as witnessed by the final invariant $F_\infty$).
\end{restatable}
The proof idea is that we are making progress whenever we add a fence, since there are only finitely many possible fence positions. Moreover, if we put a fence before every instruction, then there certainly cannot be any speculative leaks.

\begin{restatable}[Incrementality]{lemma}{refinementlemma}
  \label{lemma:refinement}
Let $P$ be a program and let $M$ and $\hat{M}$ be its transition
system and speculative execution transition system, respectively. 
For $i\in P$, if $\hat{\Init}\models\neg\fence_i$ then $M \leq_\text{sim}\hat{M}_i\leq_\text{sim}\hat{M}$.
\end{restatable}

The proof of this lemma relies on the fact that adding a fence can only 
exclude executions in $\hat{M}$ that are not possible in $M$.

By \cref{lemma:refinement}, over-approximations computed with respect to $\hat{M}$
are also over-approximations with respect to $M$.
This allows \FenceSynth to reuse information between different repair iterations.
While rule \AddFence resets $\Queue$, $\Reach$, and $\pi$, it
does not reset the current level $N$, the trace $F_0, \ldots, F_N$, and the invariant $F_\infty$,
where $F_\infty$ and $F_j$ over-approximate the states that are reachable (in up to $j$ steps).
Note that while \AddFence resets $\Reach$,
in practice parts of $\Reach$ can be retained even after a fence is added.
Intuitively, the repair loop resembles a Counterexample Guided Abstraction Refinement (CEGAR)
loop~\cite{DBLP:conf/cav/ClarkeGJLV00}.

This \emph{incremental} way of using \cref{alg:spacer} within \FenceSynth makes
repair already much more efficient than a \emph{non-incremental} version that completely re-starts verification after 
every modification (see \cref{sec:evaluation}).
However, allowing the algorithm to add fences after every instruction in $P$ 
may still be inefficient, as there are (unnecessarily) many possibilities, and the repair may add 
many unnecessary fences.
Therefore, in the following we describe several heuristics
for optimizing the way fences are added.%

\vspace{-.5em}
\subsection{Heuristics and Optimizations}
\label{sec:extensions}

\paragraph{Fence placement options}
As described in \cref{sec:spec-semantics},
speculation can only start at conditional instructions and the information leak itself happens
at an instruction $i \in \Mem$.
Thus, we can restrict the set of instructions for which we introduce
fence variables $\fence_i$ to one of the following:
\begin{inparaenum}[(i)]
\item after each conditional instruction $i \in C$, in both branches, or\label{item:branch}
\item before every instruction $i \in \Mem$.\label{item:memory}
\end{inparaenum}
In the following, option (\ref{item:branch}) will be called the \emph{after-branch}, and (\ref{item:memory}) will be called \emph{before-memory}.
In both cases correctness according to \cref{thm:correctness} is preserved.

\paragraph{Fence activation}
When a leaking execution $\pi$ is found, there might be multiple positions 
for a fence that would prevent it.%
We employ a heuristic that activates a fence variable as close as possible to the bad state. 
In case that fence placement option (\ref{item:branch}) is used, this means that $\pi$ is traversed 
backwards from the bad state until a conditional instruction is reached, and
a fence is activated in the branch that appears in $\pi$.
Under option (\ref{item:memory}), the instruction $i\in\Mem$ that is causing the
leak in $\pi$ is the last instruction in $\pi$, and we activate the fence right before $i$. 
In both cases, this not only removes the given leaking execution, but also other
executions where speculation starts before the newly added fence (cp. \cref{lemma:repair}).

%% file: evaluation.tex
\vspace{-.5em}
\section{Implementation and Evaluation}
\label{sec:evaluation}

\paragraph{Implementation}
We implemented \FenceSynth\footnote{\url{https://github.com/user-28119294/CureSpec}} in the \seahorn verification framework~\cite{DBLP:conf/cav/GurfinkelKKN15}.
This gives us direct access to an LLVM front-end. 
We compile each benchmark (see below) with Clang\footnote{\url{https://clang.llvm.org/}} 10.0.0 to LLVM using optimization level -O2.
The speculative execution semantics is added \emph{after} these compilation passes, which is important because they might introduce new vulnerabilities.%

The modified program is encoded into Horn rules~\cite{DBLP:conf/sas/BjornerMR13} and then passed to Z3~\cite{DBLP:conf/tacas/MouraB08} version 4.10.2, SeaHorn's internal solver.
Upon termination, our tool can output the inductive invariant of the repaired program, together with its speculative semantics, in SMT-LIB format, such that this certificate can be checked independently by any SMT solver that supports the LIA theory.

\paragraph{Evaluation: Benchmarks}
We evaluated \FenceSynth on four sets of benchmarks.
The first set consists of Kocher's 15 test cases\footnote{\url{https://www.paulkocher.com/doc/MicrosoftCompilerSpectreMitigation.html}}, which are simple code snippets vulnerable to \spectre attacks. %
To show that \FenceSynth can also handle complex programs from a domain that handles sensitive data, we have selected representative and non-trivial (measured in the number of LLVM instructions and conditional instructions, see \cref{tab:evaluation} for details) benchmarks from OpenSSL 3.0.0\footnote{\url{https://github.com/openssl/openssl/tree/openssl-3.0.0}} and the HACL*~\cite{DBLP:conf/ccs/ZinzindohoueBPB17} cryptographic library.

A first set of OpenSSL benchmarks includes two versions of AES block encryptions:
\begin{inparaenum}[(i)]
\item \texttt{aes\_encrypt}, which uses %
  lookup tables and\label{item:aes-lut}
\item \texttt{aes\_encrypt\_ct}, a constant-time version\label{item:aes-ct}.
\end{inparaenum}
Both of these only encrypt a single AES block. We also include both AES encryptions in cipher block chaining mode (\texttt{aes\_cbc\_encrypt} and \texttt{aes\_cbc\_encrypt\_ct}, respectively), which encrypt arbitrarily many blocks, resulting in significantly more complex and challenging benchmarks.
The second set of SSL benchmarks includes functions used for the multiplication, squaring, and exponentiation of 
arbitrary-size integers (\texttt{bn\_mul\_part}, \texttt{bn\_sqr\_part}, and \texttt{bn\_exp\_part}, respectively).
Since the full versions of these include function calls to complex subprocedures (with an additional 18900 LLVM-instructions and more than 4100 branches), we abstract these called functions by uninterpreted functions. %
Finally, the HACL* benchmarks include the following cryptographic primitives:
\begin{inparaenum}[(i)]
\item \texttt{Curve25519\_64\_ecdh}, the ECDH key agreement using Curve25519~\cite{DBLP:conf/pkc/Bernstein06},
\item the stream cipher \texttt{Chacha20\_encrypt},
\item a message authentication code using the Poly1305 hash family~\cite{DBLP:conf/fse/Bernstein05} (\texttt{Poly1305\_32\_mac}).
\end{inparaenum}

\paragraph{Evaluation: Repair Performance}
We present in detail the repair performance of \FenceSynth based on the strong threat model (see \cref{sec:spec-semantics}),
i.e., \Mem consists of all memory accesses in the analyzed program\footnote{We did not analyze the benchmarks with respect to the classical \spectre model, since that requires manual annotations of the code to determine variables that the attacker can control, which require a deep understanding of the code to be analyzed.}, and on the semantics with unbounded speculation. 
Similar results for a bounded speculation window can be found in \cref{sec:appendix-evaluation}.
For each program, we evaluated the performance for all combinations of the following options:
\begin{inparaenum}[(a)]
\item incremental or non-incremental repair,
\item fence placement at every instruction, or according to one of the heuristics (after-branch or before-memory) from \cref{sec:extensions}.%
\end{inparaenum}
All experiments were executed on an Intel\textsuperscript{®} Xeon\textsuperscript{®} Gold 6244 CPU @ 3.60GHz with 251GiB of main memory.

\begin{table}[tb]
\vspace{-1em}
  \setlength{\tabcolsep}{3pt}
  \centering
  \caption{\footnotesize Evaluation results on Kocher's test cases, OpenSSL, and HACL* benchmarks.
	Test cases not in the table have the similar values as test2.
    Columns $\#_i$, $\#_b$, and $\#_m$ give the number of instructions, conditional instructions and memory instructions, respectively, which is also the maximal number of possible fences for the respective placement option (every-inst, after-branch, before-memory).
    The number of inserted fences is shown in the $\#_f$ columns.
    Presented times are in seconds, averaged over 3 runs, with a timeout of 2 hours.
    RSS is the average maximum resident set size in GiB.
  }
  \label{tab:evaluation}
  \resizebox{1 \linewidth}{!}{
    \input{avg_all_compare}
  }
\end{table}

\cref{tab:evaluation} summarizes the most important results, comparing the \textbf{baseline}
(non-incremental, every-inst fence placement without fence activation heuristic) to two options that both use incrementality and the 
heuristic for fence activation, as well as either \textbf{after-branch} or \textbf{before-memory} for fence placement. 
We observe that the latter two perform
better than the baseline option, with a signi\-ficant difference both in repair time and number of 
activated fences, even on the test cases (except for test5, which is solved with a single fence even without heuristics). 
Notably, all of the benchmarks that time out in the baseline version can be solved by at least one of the other versions, often in under 2 minutes.

While our heuristics make a big difference when compared against the baseline, a comparison between the two fence placement heuristics does not show a clear winner.
The before-memory heuristic
results in the smallest number of activated fences for all benchmarks, but the difference is usually not big.
On the other hand, for the HACL* benchmarks (and some others), this heuristic needs more time than the after-branch heuristic, and even times out for \texttt{Curver25519\_64\_ecdh}, while after-branch solves all of our benchmarks.

Furthermore, we observe that the fence placement heuristics (after-branch or before-memory) reduce the repair time by $81.8\%$ or $97.5\%$, respectively.
Incrementality, when considered over all parameter settings, reduces the repair time on average by $10.3\%$,
but on the practically relevant settings (with fence placement heuristics) it reduces repair time by $21.5\%$.

The results for a bounded speculation window are comparable, except that this seems to be significantly more challenging for \FenceSynth: repair times increase, in some cases drastically, and \texttt{aes\_cbc\_encrypt} times out regardless of the selected options.
Note however that \FenceSynth currently does not implement any optimizations that are specific to the bounded speculation mode.

Overall, our results show that \FenceSynth can repair complex code such as the OpenSSL and HACL* examples by inserting only a few fences in the right places. 
Note that for the OpenSSL and HACL* functions, the fences inserted by our repair point to \emph{possible} vulnerabilities, 
but they might not correspond to actual attacks because of over-approximations in the strong threat model.
I.e., \FenceSynth may add fences that are not strictly necessary to secure the program.

\paragraph{Evaluation: Performance of Repaired Programs}
We evaluate the performance impact of inserted fences on an Intel Core™ i7-8565U CPU @1.80GHz by comparing the runtimes of \emph{non-fixed} programs to those of repaired programs obtained with after-branch and before-memory fence placements, respectively.
To get meaningful results we only use benchmarks that have a non-negligible runtime and use 50 random seeds\footnote{obtained from \url{https://www.random.org/}} to generate input data.\footnote{Even though \FenceSynth with bounded speculation window inserts fewer fences, their impact on performance is very similar. Therefore, we only give a single table.}

\begin{table}[tb]
  \vspace{-1em}
  \setlength{\tabcolsep}{5pt}
  \centering
  \caption{\footnotesize Impact of inserted fences on performance. The time for the non-fixed program is the combined runtime (for 50 random seeds), and the impact columns show the increase in runtime compared to the non-fixed program, for programs repaired with after-branch and before-memory placements, respectively.}
  \label{tab:performance_impact}
	\scalebox{.9}{
  \input{all_performance}
	}
\end{table}

Results are summarized in \Cref{tab:performance_impact}. 
Note that the performance impact is negligible on the constant-time version of AES in cipher block chaining mode, 
while we have big performance impact on the non-constant time version, 
even though the number of instructions is similar, and the number of branches and added fences (using a given heuristic) is the same.
Moreover, in the Chacha20-implementation, the impact of the three fences from the after-branch heuristic is smaller than the impact of the two fences from the before-memory heuristic.

In summary, the number of fences does not seem to have a strong correlation with the performance impact, 
which seems to depend more on \emph{where} the fences are added, and on properties of $P$ that are not reflected in the number of instructions or branches.
For example, a fence that is placed on an error-handling branch will have much less of an impact on the overall performance than a fence that is placed on a branch that is regularly used in non-erroneous executions of the program.
Therefore, we think that optimal placement of fences with respect to their performance impact will be an important direction of future research.%

%% file: avg_all_compare.tex
\begin{tabular}{l|rrrr|rrrr|rrrr}
\toprule
\textbf{Benchmark} & \multicolumn{4}{|c}{\textbf{baseline}} & \multicolumn{4}{|c}{\textbf{after-branch}} & \multicolumn{4}{|c}{\textbf{before-memory}}\\
 & $\#_i$ & $\#_f$ & time & RSS & $\#_b$ & $\#_f$ & time & RSS & $\#_m$ & $\#_f$ & time & RSS\\
\midrule
test1 & 14 & 5 & 0.4 & 0.12 & 2 & 1 & 0.1 & 0.07 & 2 & 1 & 0.1 & 0.06\\
test2 (and others) & 18 & 7 & 0.7 & 0.15 & 2 & 1 & 0.1 & 0.07 & 4 & 1 & 0.1 & 0.07\\
test5 & 20 & 1 & 0.3 & 0.16 & 4 & 2 & 0.2 & 0.08 & 2 & 1 & 0.1 & 0.07\\
test7 & 23 & 11 & 1.5 & 0.17 & 4 & 2 & 0.2 & 0.07 & 6 & 2 & 0.2 & 0.08\\
test9 & 21 & 10 & 1.2 & 0.16 & 2 & 1 & 0.1 & 0.06 & 5 & 1 & 0.1 & 0.08\\
test10 & 20 & 15 & 1.4 & 0.15 & 4 & 2 & 0.2 & 0.07 & 4 & 2 & 0.1 & 0.07\\
test12 & 20 & 9 & 0.9 & 0.15 & 2 & 1 & 0.1 & 0.06 & 4 & 1 & 0.1 & 0.07\\
test15 & 19 & 8 & 0.9 & 0.15 & 2 & 1 & 0.1 & 0.07 & 5 & 1 & 0.1 & 0.08\\
\midrule
aes\_encrypt\_ct & 566 & \multicolumn{3}{c|}{timeout} & 4 & 3 & 0.8 & 0.09 & 38 & 3 & 1.7 & 0.32\\
aes\_encrypt & 476 & 116 & 3712.1 & 11.04 & 4 & 3 & 0.7 & 0.10 & 97 & 3 & 4.7 & 0.98\\
aes\_cbc\_encrypt\_ct & 1302 & \multicolumn{3}{c|}{timeout} & 40 & 15 & 100.7 & 0.65 & 102 & 10 & 80.6 & 1.41\\
aes\_cbc\_encrypt & 1122 & \multicolumn{3}{c|}{timeout} & 40 & 15 & 774.5 & 1.51 & 220 & 10 & 135.1 & 4.34\\
\midrule
bn\_mul\_part & 104 & 69 & 48.7 & 0.72 & 24 & 13 & 2.6 & 0.18 & 19 & 7 & 1.3 & 0.16\\
bn\_sqr\_part & 161 & 71 & 109.5 & 1.05 & 24 & 13 & 3.4 & 0.19 & 30 & 9 & 2.5 & 0.22\\
bn\_exp\_part & 307 & 139 & 1939.3 & 5.60 & 74 & 29 & 57.4 & 0.82 & 49 & 18 & 21.8 & 0.61\\
\midrule
Chacha20\_encrypt & 3552 & \multicolumn{3}{c|}{timeout} & 6 & 3 & 8.0 & 0.33 & 117 & 2 & 8.4 & 1.90\\
Poly1305\_32\_mac & 483 & 93 & 1143.3 & 4.76 & 6 & 4 & 2.2 & 0.12 & 90 & 4 & 6.4 & 0.72\\
Curve25519\_64\_ecdh & 351 & \multicolumn{3}{c|}{timeout}  & 8 & 3 & 3115.3 & 1.92 & 44 & \multicolumn{3}{c}{timeout}\\
\bottomrule
\end{tabular}

%% file: all_performance.tex
\begin{tabular}{l|r|r|r|r}
\toprule
\textbf{Benchmark} & \textbf{input size} & \textbf{non-fixed} & \textbf{after-branch} & \textbf{before-memory}\\
                   & [MiB] & time [s] & impact & impact\\
\midrule
aes\_cbc\_encrypt\_ct & 16 & 64.49 & 1.91\% & 2.26\% \\
aes\_cbc\_encrypt & 16 & 5.84 & 53.06\% & 51.55\% \\
Chacha20\_encrypt & 64 & 9.86 & 6.69\% & 7.57\% \\
Poly1305\_32\_mac & 256 & 10.07 & 67.89\% & 49.73\% \\
\bottomrule
\end{tabular}

%% file: related-work.tex
\section{Related Work and Conclusions}
\label{sec:related work}

\paragraph{Related Work}
Existing formal methods for detecting speculative information leaks usually require the program to be repaired manually.
This includes the approach by Cauligi et al.~\cite{DBLP:conf/pldi/CauligiDGTSRB20}, which explicitly models the reorder buffer and the processor pipeline, potentially achieving a higher precision than our over-approximating approach.
Similarly, the technique developed by Cheang et al.~\cite{DBLP:conf/csfw/CheangRSS19} as well as the SPECTECTOR technique~\cite{Spectector} are based on extensions of standard notions like observational determinism to speculative execution semantics, and check for these precisely.
Moreover, Haunted RelSE~\cite{DBLP:conf/ndss/DanielBR21} extends symbolic execution to reason about standard and speculative executions at the same time.
In \cite{DBLP:conf/sp/LeonK22} speculative execution and attacker capabilities are axiomatically modeled in the CAT language for weak memory models,
allowing for easy adaption to new Spectre variants.
However, it requires to unroll the program and thus has the drawback of not handling unbounded loops/recursion.

On the other hand, there are approaches that automatically repair a given program, but cannot give a formal security guarantee.
This includes SpecFuzz~\cite{SpecFuzz}, which uses fuzzing to detect and repair out-of-bounds array accesses under speculation, as well as
oo7~\cite{oo7}, which detects and repairs Spectre leaks by static analysis and a taint tracking approach on the binary level.
However, it cannot give a security guarantee since its binary-level analysis may be incomplete.

Another line of work that resembles our approach is the automatic insertion of fences
in weak memory models~\cite{DBLP:conf/fmcad/KupersteinVY10,DBLP:conf/tacas/AbdullaACLR12,DBLP:conf/esop/BouajjaniDM13}. In contrast to these approaches, our algorithm is tightly coupled with
the model checker, and does not use it as a black-box.  
\FenceSynth allows \spacer to maintain most of its state when discovering a counterexample, 
and to resume its operation after adding a fence.

Finally, Blade~\cite{Blade} implements a type-based approach to repair Spectre leaks.
The typing rules construct a dataflow graph between expressions, similar to taint tracking, and use it to detect possible information leaks.
While this approach supports automatic repair and comes with a formal security guarantee, it suffers from the usual drawbacks of type-based approaches: the typing rules assume a fixed threat model, and any change to the type system requires to manually prove correctness of the resulting type system.
In contrast, our approach is parameterized in the threat model and can easily be combined with different techniques that detect the set of vulnerable instructions.

\paragraph{Conclusions}
We present \FenceSynth, an automatic repair algorithm for information leaks
that are due to speculative execution, parametric w.r.t. the threat model.
It is implemented in the \seahorn verification framework and can handle C programs. 
When \FenceSynth detects a leak, it
repairs it by inserting a fence. This procedure is executed iteratively until
the program is proved secure. To this end, \FenceSynth uses the model checking
algorithm \Pdr incrementally, maintaining \Pdr's state between different iterations.
This allows \FenceSynth to handle realistic programs, as shown by the experimental evaluation
on various C functions from OpenSSL and HACL*.
\FenceSynth also returns an inductive invariant that enables a simple correctness check of the repair in any SMT solver.

%% file: appendix-proofs.tex
\section{Proofs}\label{sec:appendix-proofs}

\simulationlemma*
\begin{proof}
  Follows from the proof of \cref{lemma:refinement}.
\end{proof}

\gspeclemma*
\begin{proof}
    From $\sigma_0 \models \sInit$ follows (i).
    Since $\sigma_k \models \sBad$ we have $\sigma_k \models \pc = \bot$.
    In order for $\pc$ to get $\bot$, there needs to be a point where $\spec \neq 0$ holds (see $\hat{\tau}_{a_i}$).
    Since $\spec$ can never decrease starting from $0$ in $\sigma_0$, $\spec > 0$ at that point and (iii) follows.
    (ii) holds because $\tau_\bot$ never changes $\spec$ afterwards.
\end{proof}

\soundnesslemma*
\begin{proof}
We have already argued that, under the assumption that P is constant-time secure without speculation, violations of constant-time security under speculation are only possible at instructions in $\Mem$ that are reached under speculation. 
Since this is exactly the definition of $\sBad$, the statement follows immediately.
\end{proof}

\repairlemma*
\begin{proof}
  Let $i \in \pi^k$ and $j$ such that $(\sigma_{j-1}, \sigma_j') \models \hat{\tau}_i$ (see \cref{cor:tr_inst}).
  Assume $\sInit \models \fence_i$.
  Since the values of fences never change during execution, we have $\sigma_{j-1} \models \fence_i$.
  Moreover, $\sigma_{j-1} \models \spec > 0$ and thus, $\sigma_{j-1} = \sigma_j = \sigma_N$.
  We have a contradiction because $\sigma_{j-1} \models \pc = i$.
  Thus, $\sInit \models \neg\fence_i$ and $\pi^k$ is not a path of $\hat{M}_i$.
  Therefore, no prefix $\hat{\pi}$ exists such that $\hat{\pi}\pi^k$ is an execution of $\hat{M}_i$.
\end{proof}

\correctnessthm*
\begin{proof}
  The \Leak and \AddFence rules are applicable at most $|\FVars|$ times each (\cref{lemma:repair}).
  Thus, termination follows from the termination of \cref{alg:spacer}.
  After the final application of rule \AddFence, \FenceSynth analyzes $\hat{M}_s$ and constructs an inductive invariant $F_\infty$ showing that $\langle \hat{M}_s, \sBad \rangle$ is SAFE.
\end{proof}

\refinementlemma*
\begin{proof}
  Let $\pi$ be an execution of $M$.
  Since it does not involve speculative execution there exists a corresponding execution in $\hat{M}_i$ because fences only affect speculative executions.
  This shows $M \leq_\text{sim}\hat{M}_i$.
  Morever, the additional fence in $\hat{M}_i$ only removes valid executions from $\hat{M}$.
  So, $\hat{M}_i\leq_\text{sim}\hat{M}$ holds, too.
\end{proof}

%% file: appendix-evaluation.tex
\section{Additional Experimental Results}\label{sec:appendix-evaluation}

We give 4 tables that contain experimental results for all our parameter settings (except for non-incremental without any heuristic, which is the baseline from \cref{tab:evaluation}).

\begin{table}[h!]
  \setlength{\tabcolsep}{3pt}
  \centering
  \caption{incremental solving, with fence activation heuristic}
  \label{tab:evaluation_opt_incremental}
  \resizebox{1 \linewidth}{!}{
    \input{avg_all_incremental}
  }
\end{table}

\begin{table}[h!]
  \setlength{\tabcolsep}{3pt}
  \centering
  \caption{non-incremental solving, with fence activation heuristic}
  \resizebox{1 \linewidth}{!}{
    \input{avg_all_non-incremental}
  }
\end{table}

\begin{table}[h!]
  \setlength{\tabcolsep}{3pt}
  \centering
  \caption{incremental solving, with fence activation heuristic, speculation window of size 20}
  \label{tab:evaluation_window_incremental}
  \resizebox{1 \linewidth}{!}{
    \input{avg_all_window_20_incremental}
  }
\end{table}

\begin{table}[h!]
  \setlength{\tabcolsep}{3pt}
  \centering
  \caption{non-incremental solving, with fence activation heuristic, speculation window of size 20}
  \resizebox{1 \linewidth}{!}{
    \input{avg_all_window_20_non-incremental}
  }
\end{table}

%% file: avg_all_incremental.tex
\begin{tabular}{l|rrr|rrr|rrr}
\toprule
\textbf{Benchmark} & \multicolumn{3}{|c}{\textbf{every-inst}} & \multicolumn{3}{|c}{\textbf{after-branch}} & \multicolumn{3}{|c}{\textbf{before-memory}}\\
 & $\#_f$ & time & RSS & $\#_f$ & time & RSS & $\#_f$ & time & RSS\\
\midrule
aes\_encrypt\_ct & 3 & 77.7 & 7.05 & 3 & 0.8 & 0.09 & 3 & 1.7 & 0.32\\
aes\_encrypt & 3 & 99.4 & 7.90 & 3 & 0.7 & 0.10 & 3 & 4.7 & 0.98\\
aes\_cbc\_encrypt\_ct & 10 & 1300.3 & 30.53 & 15 & 100.7 & 0.65 & 10 & 80.6 & 1.41\\
aes\_cbc\_encrypt & 10 & 2143.1 & 34.35 & 15 & 774.5 & 1.51 & 10 & 135.1 & 4.34\\
\midrule
bn\_mul\_part & 7 & 6.8 & 0.76 & 13 & 2.6 & 0.18 & 7 & 1.3 & 0.16\\
bn\_sqr\_part & 9 & 15.0 & 1.06 & 13 & 3.4 & 0.19 & 9 & 2.5 & 0.22\\
bn\_exp\_part & 18 & 205.1 & 5.16 & 29 & 57.4 & 0.82 & 18 & 21.8 & 0.61\\
\midrule
Chacha20\_encrypt & 2 & 5802.9 & 81.08 & 3 & 8.0 & 0.33 & 2 & 8.4 & 1.90\\
Poly1305\_32\_mac & 4 & 53.8 & 4.70 & 4 & 2.2 & 0.12 & 4 & 6.4 & 0.72\\
Curve25519\_64\_ecdh & \multicolumn{3}{c|}{timeout} & 3 & 3115.3 & 1.92 & \multicolumn{3}{c}{timeout}\\
\bottomrule
\end{tabular}

%% file: avg_all_non-incremental.tex
\begin{tabular}{l|rrr|rrr|rrr}
\toprule
\textbf{Benchmark} & \multicolumn{3}{|c}{\textbf{every-inst}} & \multicolumn{3}{|c}{\textbf{after-branch}} & \multicolumn{3}{|c}{\textbf{before-memory}}\\
 & $\#_f$ & time & RSS & $\#_f$ & time & RSS & $\#_f$ & time & RSS\\
\midrule
aes\_encrypt\_ct & 3 & 86.4 & 7.04 & 3 & 0.8 & 0.09 & 3 & 1.9 & 0.32\\
aes\_encrypt & 3 & 103.1 & 7.90 & 3 & 1.3 & 0.11 & 3 & 5.6 & 1.01\\
aes\_cbc\_encrypt\_ct & 10 & 1819.8 & 31.63 & 15 & 88.5 & 0.75 & 10 & 91.5 & 1.70\\
aes\_cbc\_encrypt & 10 & 2302.9 & 42.97 & 15 & 2539.8 & 2.31 & 10 & 129.9 & 5.18\\
\midrule
bn\_mul\_part & 7 & 7.5 & 0.77 & 13 & 2.9 & 0.18 & 7 & 1.4 & 0.16\\
bn\_sqr\_part & 9 & 17.6 & 1.15 & 13 & 3.7 & 0.21 & 9 & 2.8 & 0.24\\
bn\_exp\_part & 18 & 249.5 & 5.66 & 29 & 69.6 & 1.05 & 18 & 28.1 & 0.65\\
\midrule
Chacha20\_encrypt & 2 & 5506.3 & 80.46 & 3 & 9.1 & 0.33 & 2 & 11.3 & 1.93\\
Poly1305\_32\_mac & 4 & 72.9 & 4.80 & 4 & 2.6 & 0.12 & 4 & 7.4 & 0.73\\
Curve25519\_64\_ecdh & \multicolumn{3}{c|}{timeout} & 3 & 2507.9 & 1.79 & \multicolumn{3}{c}{timeout} \\
\bottomrule
\end{tabular}

%% file: avg_all_window_20_incremental.tex
\begin{tabular}{l|rrr|rrr|rrr}
\toprule
\textbf{Benchmark} & \multicolumn{3}{|c}{\textbf{every-inst}} & \multicolumn{3}{|c}{\textbf{after-branch}} & \multicolumn{3}{|c}{\textbf{before-memory}}\\
& $\#_f$ & time & RSS & $\#_f$ & time & RSS & $\#_f$ & time & RSS\\
\midrule
aes\_encrypt\_ct & 3 & 245.0 & 8.40 & 3 & 1.8 & 0.10 & 3 & 3.6 & 0.35\\
aes\_encrypt & 3 & 365.4 & 8.28 & 3 & 2.8 & 0.12 & 3 & 15.6 & 1.12\\
aes\_cbc\_encrypt\_ct & 8 & 6296.9 & 66.89 & 12 & 1384.0 & 5.16 & 9 & 1080.8 & 14.84\\
aes\_cbc\_encrypt & \multicolumn{3}{c|}{timeout} & \multicolumn{3}{c|}{timeout} & \multicolumn{3}{c}{timeout}\\
\midrule
bn\_mul\_part & 7 & 8.2 & 0.82 & 13 & 3.7 & 0.19 & 7 & 1.7 & 0.17\\
bn\_sqr\_part & 9 & 18.5 & 1.27 & 13 & 5.0 & 0.21 & 9 & 3.4 & 0.24\\
bn\_exp\_part & 18 & 2046.0 & 8.39 & 29 & 1116.0 & 2.05 & 18 & 383.5 & 1.51\\
\midrule
Chacha20\_encrypt & \multicolumn{3}{c|}{timeout} & 3 & 89.8 & 0.33 & 2 & 60.0 & 2.01\\
Poly1305\_32\_mac & 4 & 100.8 & 4.84 & 4 & 5.5 & 0.13 & 4 & 10.2 & 0.78\\
Curve25519\_64\_ecdh & \multicolumn{3}{c|}{timeout} & 3 & 3073.0 & 1.91 & \multicolumn{3}{c}{timeout}\\
\bottomrule
\end{tabular}

%% file: avg_all_window_20_non-incremental.tex
\begin{tabular}{l|rrr|rrr|rrr}
\toprule
\textbf{Benchmark} & \multicolumn{3}{|c}{\textbf{every-inst}} & \multicolumn{3}{|c}{\textbf{after-branch}} & \multicolumn{3}{|c}{\textbf{before-memory}}\\
& $\#_f$ & time & RSS & $\#_f$ & time & RSS & $\#_f$ & time & RSS\\
\midrule
aes\_encrypt\_ct & 3 & 299.0 & 8.35 & 3 & 2.9 & 0.10 & 3 & 6.7 & 0.36\\
aes\_encrypt & 3 & 426.2 & 8.13 & 3 & 3.9 & 0.12 & 3 & 24.8 & 1.17\\
aes\_cbc\_encrypt\_ct & \multicolumn{3}{c|}{timeout} & 12 & 998.9 & 4.57 & 9 & 976.8 & 13.39\\
aes\_cbc\_encrypt & \multicolumn{3}{c|}{timeout} & \multicolumn{3}{c|}{timeout} & \multicolumn{3}{c}{timeout}\\
\midrule
bn\_mul\_part & 7 & 10.0 & 0.84 & 13 & 5.0 & 0.19 & 7 & 2.0 & 0.17\\
bn\_sqr\_part & 9 & 23.4 & 1.34 & 13 & 5.9 & 0.22 & 9 & 4.0 & 0.26\\
bn\_exp\_part & 18 & 1531.5 & 9.99 & 29 & 1103.9 & 2.24 & 18 & 557.4 & 1.86\\
\midrule
Chacha20\_encrypt & \multicolumn{3}{c|}{timeout} & 3 & 77.1 & 0.33 & 2 & 58.3 & 2.10\\
Poly1305\_32\_mac & 4 & 105.6 & 4.83 & 4 & 7.5 & 0.13 & 4 & 9.7 & 0.78\\
\bottomrule
\end{tabular}